\documentclass[aps,pra,twocolumn,amsmath,amssymb]{revtex4}
\usepackage{graphicx}
\usepackage{dcolumn}
\usepackage{amsmath}
\usepackage{amssymb}
\usepackage{subfigure,amsmath,verbatim,moreverb,bm}
\usepackage{amsmath,amsthm,amssymb,url,graphicx}
\def\be{\begin{equation}}
\def\ee{\end{equation}}
\def\ber{\begin{eqnarray}}
\def\eer{\end{eqnarray}}

\def\rv{{\bf r}}
\def\Rv{{\bf R}}
\def\sv{{\bf s}}
\def\fv{{\bf f}}

\def\Po{{P_{opt}}}

\def\beq{\begin{equation}}
\def\eeq{\end{equation}}

\theoremstyle{plain}
\newtheorem{thm}{Theorem}

\newtheorem{rem}{Remark}

\theoremstyle{definition}
\theoremstyle{remark}

\newcommand\N{{\mathbb N}}

\newcommand\R{{\mathbb R}}

\newcommand\eps{\varepsilon}

\def\width1{4.5}

\DeclareMathOperator{\spt}{spt}

\begin{document}
\title{Optimal-transport formulation of electronic density-functional theory}
\author{Giuseppe Buttazzo,$^1$ Luigi De Pascale,$^2$ and Paola Gori-Giorgi$^3$}
\affiliation{$^1$Dipartimento di Matematica, Universit\`a di Pisa, Largo B. Pontecorvo 5 - 56127 Pisa, Italy\\
$^2$Dipartimento di Matematica Applicata, Universit\`a di Pisa, Via Buonarroti 1/C - 56127 Pisa, Italy\\
$^3$Department of Theoretical Chemistry and Amsterdam Center for Multiscale Modeling, FEW, Vrije Universiteit, De Boelelaan 1083, 1081HV Amsterdam, The Netherlands}
\date{\today}
\begin{abstract}
The most challenging scenario for Kohn-Sham density functional theory, that is when the electrons move relatively slowly trying to avoid each other as much as possible because of their repulsion (strong-interaction limit), is reformulated here as an optimal transport (or mass transportation theory) problem, a well established field of mathematics and economics. In practice, we show that solving the problem of finding the minimum possible internal repulsion energy for $N$ electrons in a given density $\rho(\rv)$ is equivalent to find the optimal way of transporting $N-1$ times the density $\rho$ into itself, with cost function given by the Coulomb repulsion.
We use this link to put the strong-interaction limit of density functional theory on firm grounds and to discuss the potential practical aspects of this reformulation.
\end{abstract}
\maketitle
\section{Introduction}
Electronic structure theory plays a fundamental role in many different fields: material science, chemistry and biochemistry, solid state physics, surface physics. Its goal is to solve in a reliable and computationally affordable way the many-electron problem, a complex combination of quantum mechanical and many-body effects. The most widely used technique, which achieves a reasonable compromise between accuracy and computational cost, is Kohn-Sham (KS) density functional theory (DFT) \cite{HohKoh-PR-64,KohSha-PR-65}.

Optimal transport or mass transportation theory studies the optimal transfer of masses from one location to another.  
Mass transportation theory dates back to 1781 when Monge \cite{Mon-BOOK-1781} posed the problem of finding the most economical way of moving soil from one area to another, and received a boost when Kantorovich in 1942 generalized it to what is now known as the Kantorovich dual problem \cite{Kan-DAN-42}. Optimal transport problems appear in various areas of mathematics and economics.

In this article we show that one of the most challenging scenario for KS DFT, that is when the repulsion between the electrons largely dominates over their kinetic energy, can be reformulated as an optimal transport problem. As we shall see, the potential of this link between two different well-established research areas has both formal and practical aspects.  

It is difficult to write a paper fully accessible to two different communities such as mass transportation and electronic density functional theory. In an effort towards this goal we have chosen to use for both the optimal transport and the DFT part the most commonly used notation in each case, translating from one to the other throughout the paper.
The article is organized as follows. We start in Sec.~\ref{sec_sceDFT} with a review of the motivations to study the strong-interaction limit of DFT and the challenges that this limit poses. Right after, in Sec.~\ref{sec_resDFT}, we discuss the implications for DFT of our mass transportation theory reformulation of this limit, anticipating the results that will be derived in the subsequent sections. This way, the first part of the paper is a self-contained presentation written with language that is entirely familiar to the density functional theory community.
  The mass transportation theory problem is then introduced in Sec.~\ref{sec_mass} and used in Secs.~\ref{sec_refVeeSCE}-\ref{sec_theorems} to address the strong-interaction limit of DFT and derive the results anticipated in Sec.~\ref{sec_resDFT}. Simple examples, mainly thought to illustrate the problem to the mass transportation theory community, are given in Sec.~\ref{oned}. This second part of the paper is thus mainly written in a language familiar to the optimal transport reader. The last Sec.~\ref{sec_conc} is devoted to a final discussion of the connection between these two different research areas and to conclusions and perspectives. Finally, many of the technical details are given in the Appendix.

\section{Strong interactions in DFT}
\label{sec_sceDFT}
In the formulation of Hohenberg and Kohn (HK) \cite{HohKoh-PR-64}, electronic ground-state properties are calculated by minimizing the energy functional $E[\rho]$ with respect to the particle density $\rho(\rv)$,
\be\label{EnergyFunctional}
E[\rho] = F[\rho]+\int v_{\rm ext}(\rv)\,\rho(\rv)\,d\rv,
\ee
where $v_{\rm ext}(\rv)$ is the external potential  and $F[\rho]$ is a universal functional of the density, defined as the expectation value of the internal energy (kinetic energy $\hat T=-\frac{1}{2}\sum_{i=1}^N\nabla^2_i$ plus electron-electron interaction energy $\hat V_{ee}=\sum_{i=1}^N\sum_{j=i+1}^N|\rv_i-\rv_j|^{-1}$) of the minimizing wave function that yields the density $\rho(\rv)$ \cite{Lev-PNAS-79},
\beq
F[\rho]=\min_{\Psi\to\rho}\langle\Psi|\hat{T}+\hat{V}_{ee}|\Psi\rangle.
\eeq
  Here and throughout the paper we use Hartree atomic units ($\hbar=m_e=a_0=e=1$).

In the standard Kohn-Sham approach \cite{KohSha-PR-65} the minimization of $E[\rho]$ in Eq.~\eqref{EnergyFunctional} is done under the assumption that the kinetic energy dominates over the electron-electron interaction by introducing the functional $T_s[\rho]$, corresponding to the minimum of the expectation value of  $\hat T$  alone over all fermionic (spin-$\frac{1}{2}$ particles) wave functions yielding the given $\rho$ \cite{Lev-PNAS-79},
\beq
T_s[\rho]=\min_{\Psi\to\rho}\langle\Psi|\hat{T}|\Psi\rangle.
\eeq
  The functional $T_s[\rho]$ defines a non-interacting reference system with the same density of the interacting one. The remaining part of the exact energy functional, 
\beq
E_{\rm Hxc}[\rho]\equiv F[\rho]-T_s[\rho], 
\eeq
is approximated. Usually $E_{\rm Hxc}[\rho]$ is split as $E_{\rm Hxc}[\rho]=U[\rho]+E_{xc}[\rho]$, where $U[\rho]$ is the classical Hartree functional,
\beq
U[\rho]=\frac{1}{2}\int d\rv\int d\rv'\frac{\rho(\rv)\rho(\rv')}{|\rv-\rv'|},
\label{eq_Hartree}
\eeq
 and the exchange-correlation energy $E_{xc}[\rho]$ is the crucial quantity that is approximated.  

The KS approach works well in many scenarios, but as expected, runs into difficulty where particle-particle interactions play a more prominent role. In such cases, the physics of the HK functional $F[\rho]$ is completely different than the one of the Kohn-Sham non-interacting system, so that trying to capture the difference $F[\rho]-T_s[\rho]$ with an approximate functional is a daunting task. A piece of {\em exact} information on $E_{xc}[\rho]$ is provided by the functional $V_{ee}^{\rm SCE}[\rho]$, defined as the minimum of the expectation value of $\hat{V}_{ee}$ alone over all wave functions yielding the given density $\rho(\rv)$,
\begin{equation}
	 V_{ee}^{\rm SCE}[\rho]=\min_{\Psi\to\rho}\langle\Psi|\hat{V}_{ee}|\Psi\rangle.
	\label{eq_VeeLevy}
\end{equation}
 The acronym ``SCE'' stands for ``strictly correlated electrons'' \cite{Sei-PRA-99}: $V_{ee}^{\rm SCE}[\rho]$ defines a system with maximum possible correlation between the relative electronic positions (in the density $\rho$), and it is the natural counterpart of the KS non-interacting kinetic energy $T_s[\rho]$. Its relevance for $E_{xc}[\rho]$ increases with the importance of particle-particle interactions  with respect to the kinetic energy \cite{SeiPerLev-PRA-99,SeiPerKur-PRL-00}. For low-density many-particle scenarios,
it has been shown that $V_{ee}^{\rm SCE}[\rho]$ is a much better zero-order approximation to $F[\rho]$ than $T_s[\rho]$ \cite{GorSeiVig-PRL-09,LiuBur-JCP-09,GorSei-PCCP-10}: this defines a ``SCE-DFT'' alternative and complementary to standard KS DFT. In more general cases, the dividing line between the regime where the KS approach with its current approximations works well and the regime where a SCE-based approach is more suitable is a subtle issue, with many complex systems being not well described by neither KS nor SCE (see also the discussion in Ref.~\onlinecite{GorSei-PCCP-10}).

The functional $V_{ee}^{\rm SCE}[\rho]$ also contains exact information on the important case of the stretching of the chemical bond \cite{TeaCorHel-JCP-10,GorSei-PCCP-10}, a typical situation in which restricted KS DFT encounters severe problems. The relevance of $V_{ee}^{\rm SCE}[\rho]$ for constructing a new generation of approximate $E_{xc}[\rho]$ has also been pointed out very recently by Becke \cite{Bec-cecam-11}. Notice that $V_{ee}^{\rm SCE}[\rho]$ also enters in the derivation of the Lieb-Oxford bound \cite{Lie-PLA-79,LieOxf-IJQC-81,Per-INC-91,LevPer-PRB-93,RasSeiGor-PRB-11}, an important exact condition on $E_{xc}[\rho]$.

Overall, constructing the functional $V_{ee}^{\rm SCE}[\rho]$ for a given density $\rho(\rv)$ in an exact and efficient way has the potential to extensively broaden the applicability of DFT. Only approximations for $V_{ee}^{\rm SCE}[\rho]$ were available \cite{SeiPerKur-PRA-00} until recently, when the mathematical structure of the exact $V_{ee}^{\rm SCE}[\rho]$ has been investigated in a systematic way \cite{SeiGorSav-PRA-07,GorVigSei-JCTC-09} and exact solutions for spherically-symmetric densities (which have been used in the first SCE-DFT calculations \cite{GorSeiVig-PRL-09,GorSei-PCCP-10}) have been produced. However, a general reliable algorithm to construct  $V_{ee}^{\rm SCE}[\rho]$ is still lacking, and many formal aspects still need to be addressed. Here is where mass transportation theory can play a crucial role. Reformulating $V_{ee}^{\rm SCE}[\rho]$ as an optimal transport problem allows to put the construction of this functional on firm grounds and to import algorithms from another well-established research field. 

\section{Results: an overview}
\label{sec_resDFT}

The problem posed by Eq.~\eqref{eq_VeeLevy}, that is searching for the minimum possible interaction energy in a given density,  was first addressed, in an approximate way, in the seminal work of Seidl and coworkers \cite{Sei-PRA-99,SeiPerLev-PRA-99,SeiPerKur-PRA-00}. Later on, in Refs.~\onlinecite{SeiGorSav-PRA-07} and \onlinecite{GorVigSei-JCTC-09}, a formal solution was given in the following way.
The admissible configurations of $N$ electrons in $d$ dimensions are restricted to a $d$-dimensional subspace $\Omega_0$ of the full $Nd$-dimensional configuration space. A generic point of $\Omega_0$ has the form $\Rv_{\Omega_0}(\sv) = (\fv_1(\sv),....,\fv_N(\sv))$
where $\sv$ is a $d$-dimensional vector that determines the position of, say, electron ``1", 
 and $\fv_i(\sv)$ ($i=1,...,N$, $\fv_1(\sv)=\sv$) are the {\it co-motion functions},  which determine the position of the $i$-th electron in terms of $\sv$.   The variable $\sv$ itself is distributed according to the normalized density $\rho(\sv)/N$.  The co-motion functions are implicit functionals of the density, determined by a set of differential equations that ensure the invariance of the density under the transformation $\sv \to \fv_i(\sv)$, 
\begin{equation}
\rho(\fv_i(\sv))d\fv_i(\sv)=\rho(\sv)d\sv.
\label{eq_df}
\end{equation}
They also satisfy group properties \cite{SeiGorSav-PRA-07} which ensure the indistinguishability of the $N$ electrons. The functional $V_{ee}^{\rm SCE}[\rho]$ is then given by
\beq
V_{ee}^{\rm SCE}[\rho]=\int d\sv\, \frac{\rho(\sv)}{N} \, \sum_{i=1}^N\sum_{j=i+1}^N\frac{1}{|\fv_i(\sv)-\fv_j(\sv)|}.
\eeq
Notice that while in chemistry only the three-dimensional case is interesting, in physics systems with reduced effective dimensionality (quantum dots, quantum wires, point contacts, etc.) play an important role.

As we shall see in Secs.~\ref{sec_mass}-\ref{sec_refVeeSCE}, this way of addressing the functional $V_{ee}^{\rm SCE}[\rho]$ corresponds to an attempt of solving the so-called {\em Monge problem} associated to the constrained minimization of Eq.~\eqref{eq_VeeLevy}. In the Monge problem, one essentially tries to transport a mass distribution $\rho_1(\rv)d\rv$ into a mass distribution $\rho_2(\rv)d \rv$  in {\it the most economical way} according to a given definition of the work necessary to move a unit mass from position $\rv_1$ to position $\rv_2$. 
For example, one may wish to move books from one shelf (``shelf 1'') to another (``shelf 2''), by minimizing the total work. The goal of solving the Monge problem is then to find an {\em optimal map} which assigns to every book in shelf 1 a unique final destination in shelf 2.
In Secs.~\ref{sec_mass}-\ref{sec_refVeeSCE}, it will then become clear that the co-motion functions are the optimal maps of the Monge problem associated to $V_{ee}^{\rm SCE}[\rho]$. 

However, it is well known in mass transportation theory that the Monge problem is very delicate and that proving in general the existence of the optimal maps (or co-motion functions) is extremely difficult. In 1942 Kantorovich proposed a relaxed formulation of the Monge problem, in which the goal is now to find the probability that, when minimizing the total cost, a mass element of $\rho_1$ at position $\rv_1$ be transported at position $\rv_2$ in $\rho_2$. As detailed in Sec.~\ref{sec_refVeeSCE}, this formulation is actually the appropriate one for the  constrained minimization of Eq.~\eqref{eq_VeeLevy}.

We were then able to prove in Sec.~\ref{sec_theorems} four theorems on $V_{ee}^{\rm SCE}[\rho]$. In the first one, the existence of a generalized minimizer for Eq.~\eqref{eq_VeeLevy} is rigorously established. It is useful to remind here that the functional $V_{ee}^{\rm SCE}[\rho]$ corresponds to the $\lambda\to \infty$ limit \cite{Sei-PRA-99,SeiPerLev-PRA-99} of the traditional adiabatic connection of DFT \cite{HarJon-JPF-74,Har-PRA-84,LanPer-SSC-75,Lev-PRA-91}, in which a functional $F_\lambda[\rho]$ depending on a real parameter $\lambda$ is defined as
\beq
F_{\lambda}[\rho]=\min_{\Psi\to\rho}\langle\Psi|\hat{T}+\lambda\,\hat{V}_{ee}|\Psi\rangle.
\label{eq_Flambda}
\eeq
If $\Psi_\lambda[\rho]$ is the minimizer of Eq.~\eqref{eq_Flambda}, and if we define
\beq
W_{\lambda}[\rho]\equiv\langle\Psi_\lambda[\rho]|\hat{V}_{ee}|\Psi_\lambda[\rho]\rangle-U[\rho],
\eeq
we have, under mild assumptions, the well-known exact formula \cite{LanPer-SSC-75} for the exchange-correlation functional of KS DFT:
\beq
E_{xc}[\rho]=\int_0^1W_\lambda[\rho] \, d\lambda.
\label{eq_ExcInt}
\eeq
When $\lambda\to\infty$ it can be shown  that \cite{Sei-PRA-99,SeiPerLev-PRA-99,SeiGorSav-PRA-07,GorVigSei-JCTC-09}
\beq
\lim_{\lambda\to\infty} W_\lambda[\rho]=V_{ee}^{\rm SCE}[\rho]-U[\rho],
\eeq
where $U[\rho]$ is the Hartree functional of Eq.~\eqref{eq_Hartree}. We have thus put the existence of this limit, which contains a piece of exact information that can be used to model $E_{xc}[\rho]$ \cite{Ern-CPL-96,BurErnPer-CPL-97,SeiPerLev-PRA-99,TeaCorHel-JCP-10,GorSei-PCCP-10,Bec-cecam-11},  on firm grounds.

When $\rho(\rv)$ is ground-state $v$-representable $\forall\, \lambda$, $\Psi_\lambda[\rho]$ is the ground state of the hamiltonian
\begin{equation}
\hat{H}_\lambda[\rho]  =  \hat{T}+\lambda\,\hat{V}_{ee}+\hat{V_\lambda}[\rho], 
\label{eq_Hlambda}
\end{equation}
where
\begin{equation}
\hat{V_\lambda}[\rho] =  \sum_{i=1}^N v_\lambda[\rho](\rv_i)
\end{equation}
is a one-body local potential that keeps the density equal to the physical ($\lambda=1$) $\rho(\rv)$ $\forall\; \lambda$. In Refs.~\onlinecite{SeiGorSav-PRA-07,GorVigSei-JCTC-09} and \onlinecite{GorSeiVig-PRL-09} it has been argued that
\beq
\lim_{\lambda\to\infty} \frac{v_\lambda[\rho](\rv)}{\lambda}=v_{\rm SCE}[\rho](\rv),
\eeq
where $v_{\rm SCE}[\rho](\rv)$ is related to the co-motion functions via the classical equilibrium equation \cite{SeiGorSav-PRA-07}
\beq
\nabla v_{\rm SCE}[\rho](\rv)=\sum_{i= 2}^N \frac{\rv-\fv_i(\rv)}{|\rv-\fv_i(\rv)|^3},
\label{eq_vSCEf}
\eeq
and it is the counterpart of the KS potential in the strong-interaction limit.
In fact, we also have
\beq
\frac{\delta V_{ee}^{\rm SCE}[\rho]}{\delta \rho(\rv)}=-v_{\rm SCE}[\rho](\rv).
\label{eq_VeeSCEfuncder}
\eeq
While Eq.~\eqref{eq_vSCEf} is only valid if the co-motion functions (optimal maps) exist, Eq.~\eqref{eq_VeeSCEfuncder} is more general.
As we shall see in Secs.~\ref{sec_refVeeSCE}-\ref{sec_theorems},  the Kantorovich problem can be rewritten in a useful dual formulation in which the so called {\em Kantorovich potential} $u(\rv)$ plays a central role. The relation between the Kantorovich potential and $v_{\rm SCE}[\rho](\rv)$ is simply
\beq
u(\rv)=-v_{\rm SCE}[\rho](\rv)+C,
\label{eq_reluvSCE}
\eeq
where $C$ is a constant that appears if we want to set $v_{\rm SCE}(|\rv|\to\infty)=0$, with $|\rv|$ denoting the distance from the center of charge of the external potential.
With our Theorems~\ref{dual}-\ref{thsumm} we have proved that under very mild assumptions on $\rho(\rv)$ this potential exists, it is bounded and it is differentiable almost everywhere, also for cases in which the co-motion functions do not exist, thus addressing the $v$-representability problem in the strong-interaction ($\lambda\to\infty$) limit. 

Theorem \ref{thsumm} also proves that the value of $V_{ee}^{\rm SCE}[\rho]$ is exactly given by the  maximum of the Kantorovich dual problem
\begin{eqnarray}
& & V_{ee}^{\rm SCE}[\rho]  =  \\
& & 	\max_u\left\{ \int u(\rv)\rho(\rv) d\rv \ : \ \sum_{i=1}^N u(\rv_i)\le \sum_{i=1}^N\sum_{j>i}^N \frac{1}{|\rv_i-\rv_j|} \right\}.\nonumber
\label{eq_VeeKantorovic}
\end{eqnarray}
The condition $\sum_{i=1}^N u(\rv_i)\le \sum_{i=1}^N\sum_{j>i}^N\frac{1}{|\rv_i-\rv_j|} $ has a simple physical meaning: it requires that at optimality the allowed subspace $\Omega_0$ of the full $Nd$ configuration space be a minimum of the classical potential energy. This can be easily verified by rewriting this condition in terms of $v_{\rm SCE}[\rho](\rv)$ using Eq.~\eqref{eq_reluvSCE}:
\beq
\sum_{i=1}^N\sum_{j>i}^N\frac{1}{|\rv_i-\rv_j|}+\sum_{i=1}^N v_{\rm SCE}[\rho](\rv_i)\ge E_{\rm SCE},
\eeq
where the equality is satisfied only for configurations belonging to $\Omega_0$, and $E_{\rm SCE}$ is the total energy in the SCE limit \cite{GorVigSei-JCTC-09}: $E_{\rm SCE}=\lim_{\lambda\to\infty}\lambda^{-1}E_\lambda$, where $E_\lambda$ is the ground-state energy of \eqref{eq_Hlambda}.

Equation~\eqref{eq_VeeKantorovic} is related to the Legendre transform formulation of Lieb \cite{Lie-IJQC-83} of the KS functionals, but it has the advantage of being only a maximization under linear constraints, meaning that it can be dealt with linear programming techniques.

We were not able to prove the existence of the co-motion functions (optimal maps) in the general case, although we have hints that, for reasonable densities, it migth be possible. As mentioned, this is always a delicate problem. We could only prove the existence of an optimal map in the special case $N=2$ (Appendix \ref{app_fN2}). 

In the following sections we introduce the optimal transport problem and we give the details of the results anticipated here.

\section{Optimal transport} 
\label{sec_mass}
In 1781 Gaspard Monge \cite{Mon-BOOK-1781} proposed a model to describe the work necessary to move a mass distribution $p_1=\rho_1\,dx$ into a final destination $p_2=\rho_2\,dx$, given the unitary transportation cost function $c(x,y)$ which measures the work to move a unit mass from $x$ to $y$. The goal is to find a so-called {\it optimal transportation map} $f$ which moves $p_1$ into $p_2$, i.e. such that
\begin{equation}\label{push1}
p_2(S)=p_1\big(f^{-1}(S)\big)\qquad \forall\; {\rm measurable\; sets\; }S,
\end{equation}
with minimal total transportation cost
\begin{equation}
\int c\big(x,f(x)\big)\,dp_1.
\end{equation}
The measures $p_1$ and $p_2$, which must have equal mass (normalized to one for simplicity), are called {\it marginals}. 
The natural framework for this kind of problems is the one where $X$ is a metric space and $p_1,p_2$ are probabilities on $X$. However, the existence of an optimal transport map is a very delicate question (for a simple example, see Sec.~\ref{oned}), even in the classical Monge case, where $X$ is the Euclidean space $\R^d$ and the cost function is the distance between $x$ and $y$, $c(x,y)=|x-y|$. Thus in 1942 Kantorovich \cite{Kan-DAN-42} proposed  a relaxed formulation of the Monge transport problem: the goal is now to find a probability $P(x,y)$ on the product space, which minimizes the relaxed transportation cost
$$\int c(x,y)\,P(dx,dy)$$
over all admissible probabilities $P$, where admissibility means that the projections $\pi^\#_1P$ and $\pi^\#_2P$ coincide with the marginals $p_1$ and $p_2$ respectively. Here the notation $\pi^\#_iP$ means that we integrate $P$ over all variables except the $i^{\rm th}$.
The Kantorovich problem then reads
\beq
\min_P\Big\{\int c(x,y)\,P(dx,dy)\ :\ \pi^\#_jP=p_j\hbox{ for }j=1,2\Big\},
\label{eq_kant2mar}
\eeq
where $j=1,2$ denotes, respectively, the variables $x$ and $y$. The minimizing $P(dx,dy)=P(x,y)dxdy$ in Eq.~\eqref{eq_kant2mar}, called {\em transport plan}, gives the probability that a mass element in $x$ be transported in $y$: this is evidently more general than the Monge transportation map $f$ which assigns a unique destination $y$ to each $x$.

The generalization to more than two marginals is crucial for our purpose and is written as
\begin{eqnarray}
 \min_P\Big\{\int c(x_1,\dots,x_N) \,P(dx_1,\dots,dx_N)\ : \ \nonumber \\
 \pi^\#_jP=p_j\hbox{ for }j=1,\dots,N\Big\}.
\label{optpb}
\end{eqnarray}
The analogous of the Monge problem in this case is to find $N$ maps $f_i$ such that $f_1(x)=x$, $p_i(S)=p_1\big(f_i^{-1}(S)\big)$ for every measurable set $S$, and $(f_1,\dots,f_N)$ minimizes 
$$\int c(f_1(x_1),\dots,f_N(x_1))\,p_1(dx_1),$$
among all maps with the same property. 

\section{Reformulation of $V_{ee}^{\rm SCE}[\rho]$} 
\label{sec_refVeeSCE}
We can now see that the way in which  $V_{ee}^{\rm SCE}[\rho]$ was addressed in Refs.~\onlinecite{SeiGorSav-PRA-07} and \onlinecite{GorVigSei-JCTC-09} (briefly reviewed in Sec.~\ref{sec_resDFT}) corresponds to an attempt of solving the Monge problem associated to the constrained minimization of Eq.~\eqref{eq_VeeLevy}, where the co-motion functions are the optimal maps. Indeed, Eq.~\eqref{push1} is a weak form of Eq.~\eqref{eq_df} which does not require $f$ to be differentiable.

However, as said, proving the existence of the optimal maps is in general a delicate problem. Moreover, the problem posed by Eq.~\eqref{eq_VeeLevy} has actually the more general Kantorovich form \eqref{optpb}. This can be seen by doing the following (with $x\in \R^d$):
\begin{itemize}
	\item identify the probability $P(d x_1,\dots,d x_N)$ with $|\Psi(x_1,\dots,x_N)|^2dx_1,\dots,dx_N$;
	\item set all the marginals $p_i$ equal to the density divided by the number of particles $N$, $p_i=\frac{1}{N}\rho\, d x$;
	\item set the cost function equal to the electron-electron Coulomb repulsion, \begin{equation}
	c(x_1,x_2,\dots,x_N)=\sum_{i=1}^N\sum_{j=i+1}^N\frac{1}{|x_i-x_j|} \label{eq_Coulcost}.\end{equation}
\end{itemize}
Thus, solving the problem of finding the minimum possible electron-electron repulsion energy in a given density is equivalent to find the optimal way of transporting $N-1$ times the density $\rho$ into itself, with cost function given by the Coulomb repulsion, in the relaxed Kantorovich formulation.

What are the advantages of this reformulation? As anticipated in Sec.~\ref{sec_resDFT}, we can put many of the conjectures on $V_{ee}^{\rm SCE}[\rho]$ \cite{SeiGorSav-PRA-07,GorSeiVig-PRL-09,GorSei-PCCP-10} on firm grounds, and we can rewrite Eq.~\eqref{eq_VeeLevy} in a convenient dual form that allows to use linear programming techniques, with the potential of giving access to a toolbox of algorithms already developed in a different, well-established, context.

\section{Theorems on $V_{ee}^{\rm SCE}[\rho]$} 
\label{sec_theorems}
From the point of view of mass transportation theory, the problem of Eq.~\eqref{eq_VeeLevy} poses two challenges: {\it i}) the cost function corresponding to the Coulomb potential, Eq.~\eqref{eq_Coulcost}, is different from the usual cost considered in the field. In particular, it is not bounded at the origin and it decreases with distance, thus requiring a generalized formal framework; {\it ii}) the literature on the problem with several marginals is not very extensive (see, e.g., \cite{GanSwi-CPAM-98,RacRus-BOOK-98}). Nonetheless, we could prove several results. In what follows, we state them, relegating many technical details of the proofs in the Appendix. 
\begin{thm}\label{exis}
If the cost function $c$ is nonnegative and lower semicontinuous there exists an optimal probability $P_{opt}$ for the minimum problem \eqref{optpb}.
\end{thm}
\begin{proof}
The proof is an application of the Prokhorov compactness theorem for measures. In fact, taken a minimizing sequence $(P_n)$ for problem \eqref{optpb}, since they are all probabilities, the sequence $(P_n)$ is compact in the weak* convergence of measures, so (a subsequence of) it converges weakly* to a nonnegative measure $P$, and this is enough to obtain
$$\int c\,dP\le\liminf_n\int c\,dP_n.$$
Then $P$ is a good candidate for being an optimal probability for problem \eqref{optpb}. To achieve the proof it remains to show that $P$ is a probability and that the marginal condition $\pi^\#_jP=p_j$ is fulfilled. This is true if the convergence of $(P_n)$ to $P$ is ``narrow'', which by Prokhorov theorem amounts to show the so-called tightness condition:
$$\forall\eps>0\ \exists K\hbox{ compact in $\R^{Nd}$ :}\;P_n(\R^{Nd}\setminus K)<\eps,\ \forall n\in\N.$$
The tightness condition above follows easily by the fact that all $P_n$ satisfy the 
marginal conditions $\pi^\#_jP_n=p_j$ ($j=1,\dots,N)$.
\end{proof}
\begin{rem}\label{remsim}{\rm
If the marginals $p_1,\dots,p_N$ are all equal and if the cost function $c$ satisfies the symmetry condition
\begin{equation}\label{simm}
c(x_1,\dots,x_N)=c(x_{k_1},\dots,x_{k_N})
\end{equation}
for all permutations $k$, then the existence theorem above holds with $P_{opt}$ which satisfies the same symmetry condition. In fact, it is enough to notice that, taken a probability $P$, the new probability
$$\tilde P(x_1,\dots,x_N)=\frac{1}{N!}\sum_kP(x_{k_1},\dots,x_{k_N}),$$
where $k$ runs over all permutations of $\{1,\dots,N\}$, has a cost less than or equal to the one of $P$ and the same marginals.
}\end{rem}
We now turn to the important dual reformulation. The standard dual problem in optimal transport theory is:
\begin{thm}\label{dual}
Let $c$ be a lower semicontinuous and finite valued function, then 
\begin{eqnarray}
 & &  \min_P\biggl\{\int c(x_1,\dots,x_N)\,  P(dx_1,\dots,dx_N)\, :\nonumber  \\
 & & \qquad \qquad \qquad \qquad \qquad \pi^\#_jP=p_j\hbox{ \rm for }j=1,\dots,N\biggr\}\nonumber \\
 & &   =\max_{u_j}\biggl\{\sum_{j=1}^N\int u_j\,dp_j  : \; \sum_{j=1}^Nu_j(x_j)\le c(x_1,\dots,x_N)\biggr\}. \nonumber
\end{eqnarray}
Moreover, the dual maximization problem also admits a solution.
\end{thm}

\begin{rem}\label{remsim2}{\rm
Again, if $p_1=\dots=p_N=p$ and if the cost function $c$ satisfies the symmetry condition \eqref{simm}, then the dual problem admits a solution $u_1=\dots=u_N=u$. In fact, if $u_1,\dots,u_N$ is an optimal solution of the dual problem, the function
$$u(x)=\frac{1}{N}\big(u_1(x)+\dots+u_N(x)\big)$$
has the same maximal dual cost, and satisfies the constraint
$$u(x_1)+\dots+u(x_N)\le c(x_1,\dots,x_N).$$
Therefore, in this situation the dual problem becomes
\begin{equation}\label{dualpb}
\max_u\Big\{N\int u\,dp\ :\ 
\sum_{i=1}^N u(x_i)\le c(x_1,\dots,x_N)\Big\}.
\end{equation}
An optimal function $u$ for the dual problem \eqref{dualpb} is called {\it Kantorovich potential}.
}\end{rem}

However, the theorem above does not apply directly to the optimal transport problem of interest here, because the cost, given by Eq.~\eqref{eq_Coulcost}, takes the value $+\infty$ on the set $\{x_i=x_j\hbox{ for some }i\ne j\}$. The dual formulation then takes the following aspect (see for instance \cite{Sch-arxiv-11}).

\begin{thm}\label{dualbis} 
Let $c$ be a Borel function with values in $[0,+\infty]$ and assume that
$c$ is $p_1\otimes\dots\otimes p_N$ almost everywhere finite. Assume moreover that there exists a finite cost transport plan $\hat P$. Then there exists Borel measurable dual maximizers $u_i$ with values in $[-\infty, +\infty)$ such that
\begin{eqnarray*}
  & & \min_P\biggl\{\int c(x_1,\dots,x_N)\,   P(dx_1,\dots,dx_N) \; : \\ 
& & \qquad \qquad \qquad \qquad \qquad \pi^\#_jP=p_j \;{\rm for} j=1,\dots,N\biggr\} \\
 & &  =\max_{u_j}\biggl\{\int\sum_{j=1}^N u_j(x_j)\,  d{\hat P}  (dx_1,\dots,dx_N)\; : \\
  &  & \qquad \qquad \qquad \qquad \qquad  \sum_{j=1}^Nu_j(x_j)\le c(x_1,\dots,x_N)\biggr\}.
\end{eqnarray*}
\end{thm}

\noindent The assumption that $c$ is a Borel function is large enough to include continuous and lower semicontinuous functions (also taking the value $+\infty$), in particular the Coulomb potential in Eq.~\eqref{eq_Coulcost}.

The dual form of Theorem \ref{dualbis} does not allow explicit computations since it involves a plan $\hat P$ which may be not explicitly known. To overcome this difficulty we were able to prove that, for the cost \eqref{eq_Coulcost} under consideration, the more useful dual form \eqref{dualpb} still holds:

\begin{thm}\label{thsumm}
Let $c$ be the cost \eqref{eq_Coulcost} and assume all marginal measures $p_j$ coincide. Then there exists a maximizer $u$ for the dual problem of Theorem \ref{dualbis} which satisfies the formula
$$u(x)=\inf_{y_i}\Big\{c(x,y_1\dots,y_{N-1})-\sum_{i=1}^{N-1}u(y_i)\ :\ y_i\in\R^d\Big\}.$$
Such Kantorovich potential $u$ is also bounded and verifies the equality
$$\int\sum_{j=1}^N u(x_j)\,d{\hat P}(dx_1,\dots,dx_N)=N\int u(x)\,dp(x).$$
Moreover, if $p=\frac{1}{N}\rho(x)\,dx$, then $u$ is differentiable almost everywhere and $\nabla u$ is locally bounded.
\end{thm}

In Sec.~\ref{sec_resDFT} we have already discussed the physical meaning of the Kantorovich potential $u$: it is an effective
single particle potential, playing the same role of the KS potential in the strong-interaction limit. 

The proof of Theorem \ref{thsumm} is discussed in Appendix~\ref{app_proof}. We were also able to prove, as reported in Appendix \ref{app_fN2}, the existence of an optimal map ({\em co-motion function}) $f$ in the special case $N=2$, in any dimension $d$.
In the following section we show some explicit computations for simple cases.

\section{Analytical examples}\label{oned}
The purpose of this section is to illustrate the optimal transport reformulation of the strictly correlated electron problem using simple examples. Results similar to those reported here have been already obtained from physical considerations in Refs.~\onlinecite{Sei-PRA-99,SeiGorSav-PRA-07,GorSei-PCCP-10} and \onlinecite{RasSeiGor-PRB-11}, where solutions using chemical and physical densities have been presented and discussed. In a way, this section is mainly addressed to the mass transportation community, with examples of the SCE problem translated in their familiar language. The DFT reader can also gain insight about the mass transportation formulation of the SCE problem from these examples by comparing them with those of Refs.~\onlinecite{Sei-PRA-99,SeiGorSav-PRA-07,GorSei-PCCP-10} and \onlinecite{RasSeiGor-PRB-11}. 

We first consider the radial problem for two particles in a given dimension $d$, and then the case of $N$ particles in $d=1$ dimension.

\subsection{The radial $d$-dimensional case for $N=2$}\label{radial}

Here we deal with the radial case $\rho(x)=\rho(|x|)$ when the number $N$ of particles is two.

The mass density $\rho(|x|)$ is transported on itself in an optimal way by a transport map $f$ whose existence has been proved in Appendix \ref{app_fN2}. According to the one-dimensional calculations of the next subsection, for every half-line starting from the origin the mass density $r^{d-1}\rho(r)$ is transported on the opposite half-line in an optimal way. In other words we have
$$f(x)=-\frac{x}{|x|}a(|x|)$$
where the function $a(r)$ can be computed by solving the ordinary differential equation (ODE)
$$a'(r)\big(a(r)\big)^{d-1}\rho\big(a(r)\big)=-r^{d-1}\rho(r)$$
which gives
$$\int_0^{a(r)}s^{d-1}\rho(s)\,ds=\frac{1}{d\omega_d}-\int_0^rs^{d-1}\rho(s)\,ds$$
being $\omega_d$ the $d$-volume of the unit ball in $\R^d$. The Kantorovich potential $u(r)$ is obtained differentiating the dual relation $u(x)+u(y)=1/|x-y|$ at the optimal points, which gives
$$u(r)=-\int_0^r\frac{1}{(s+a(s))^2}\,ds+\frac{1}{2}\int_0^{+\infty}\frac{1}{(s+a(s))^2}\,ds\;.$$
For instance, if $d=2$ and $\rho(r)$ is the Gaussian function $\rho(r)=ke^{-kr^2}/\pi$ we find
$$a(r)=\sqrt{-\frac{1}{k}\log(1-e^{-kr^2})}.$$
Notice that these results were already obtained from physical arguments by Seidl \cite{Sei-PRA-99} in his first paper on strictly correlated electrons.

It must be also noticed that replacing the Coulomb repulsion $1/|x-y|$ by the more moderate repulsion ({\it harmonic interaction}) $-|x-y|^2/2$, similar calculations give, due to the concavity of the cost function,
$$f(x)=-x,\quad\hbox{with Kantorovich potential }u(r)=-r^2,$$
as it was already discussed in the appendix of Ref.~\onlinecite{SeiGorSav-PRA-07}.

\subsection{The case $N=2$ and $d=1$ dimension}\label{onedim2}

We take $N=2$ particles in one dimension and we first consider the simple case
\begin{equation}
\rho_1(x)=\rho_2(x)=\left\{
\begin{array}{ll}
a\quad\hbox{if }|x|\le a/2\\
0\quad\hbox{otherwise}
\end{array}\right.
\label{eq_rhosquarewell}
\end{equation}
and
$$c(x,y)=\frac{1}{|x-y|}\;.$$
By symmetry, the goal is to send the interval $[0,a/2]$ into $[-a/2,0]$ by a transportation map $f$ with minimal cost
$$F(f)=a\int_0^{a/2}\frac{1}{x-f(x)}\,dx.$$
Since the function $t\mapsto1/t$ is convex on $\R^+$, by Jensen inequality we have
$$F(f)\ge\frac{a^3}{4}\left(\int_0^{a/2}x-f(x)\,dx\right)^{-1}.$$
Taking into account that $\int_0^{a/2}x\,dx=a^2/8$ and
$$\int_0^{a/2}f(x)\,dx=\int_{-a/2}^0y\,dy=-\frac{a^2}{8}$$
we obtain that $F(f)\ge a$ for every transport map $f$. Choosing
$$f(x)=x-\frac{a}{2}$$
we have $F(f)=a$ which shows that $f$ is optimal. The plot of the optimal map $f$ on $[-a/2,a/2]$ is shown in Fig.~\ref{fig_example1}. This is the same optimal map used in Ref.~\onlinecite{RasSeiGor-PRB-11}.

\begin{figure}[h]
\begin{center}
\includegraphics[width=8.0cm]{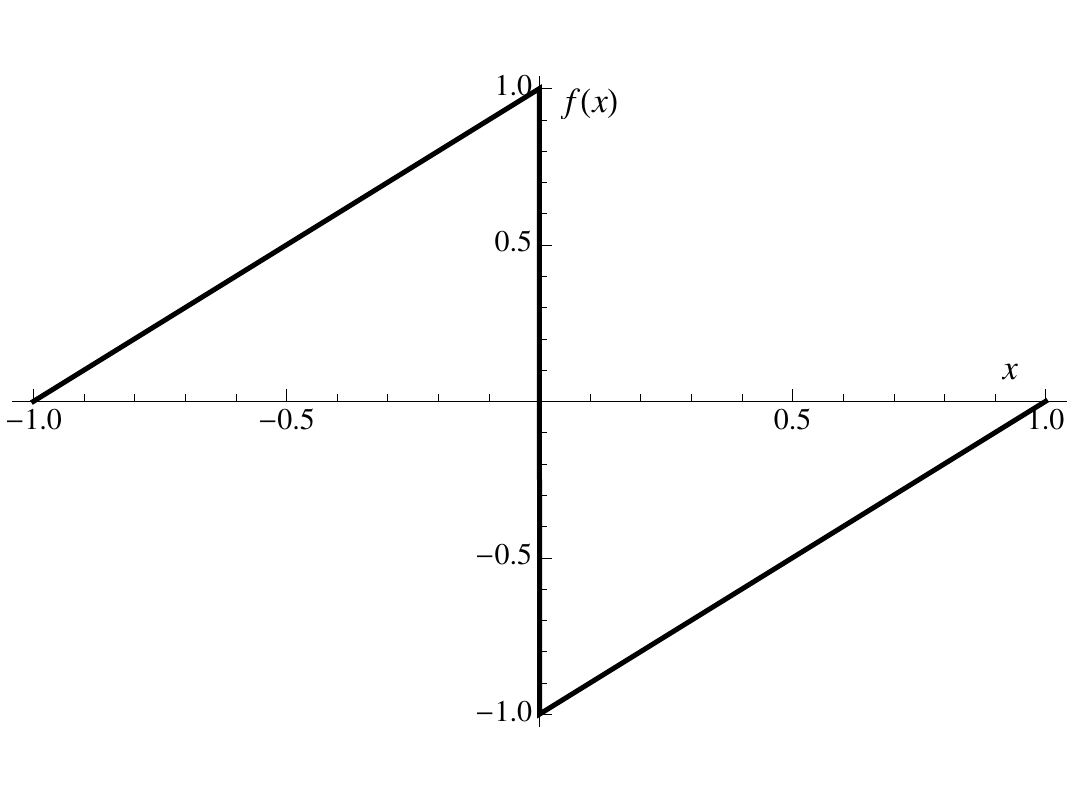}
\end{center}
\caption{The optimal map $f$ for the density of Eq.~\eqref{eq_rhosquarewell} with $a=2$.}
\label{fig_example1}
\end{figure}

Similar computations can be made for different densities $\rho$. Let us denote by $r_1$ the ``first half'' of $\rho$ and by $r_2$ the ``second half''; there is no loss of generality if we assume that the point where $\rho$ splits is the origin. In other words,
\begin{eqnarray*}
r_1 & = & \rho\hbox{ on }]-\infty,0[,\ r_2=\rho\hbox{ on }]0,+\infty[,\\
& & \hbox{ with }  \int_{-\infty}^0r_1\,dx=\int_0^{+\infty}r_2\,dx=1/2.
\end{eqnarray*}
The best transport map $f$ sends $r_1$ onto $r_2$, so from the differential relation
$$f'(x)r_2\big(f(x)\big)=r_1(x),$$
taking into account that $f(-\infty)=0$, we find
$$f(x)=R_2^{-1}\Big(R_1(x)+\frac{1}{2}\Big)\quad\hbox{for }x<0$$
where $R_1$ and $R_2$ are the two primitives of $r_1$ and $r_2$ respectively, vanishing at the origin. Analogously, we obtain
$$f(x)=R_1^{-1}\Big(R_2(x)-\frac{1}{2}\Big)\quad\hbox{for }x>0,$$
which agrees with the results of Refs.~\onlinecite{Sei-PRA-99} and \onlinecite{SeiGorSav-PRA-07}.
For instance, if
\beq
\rho(x)=\frac{a-|x|}{a^2}\qquad\hbox{defined in }[-a,a]
\label{eq_denstriang}
\eeq
we get
$$f(x)=\frac{x}{|x|}\left(\sqrt{2a|x|-x^2}-a\right)\qquad\hbox{on }[-a,a]$$
plotted in Fig.~\ref{fig_ftriangN2} for $a=1$.

\begin{figure}[h]
\begin{center}
\includegraphics[width=8.0cm]{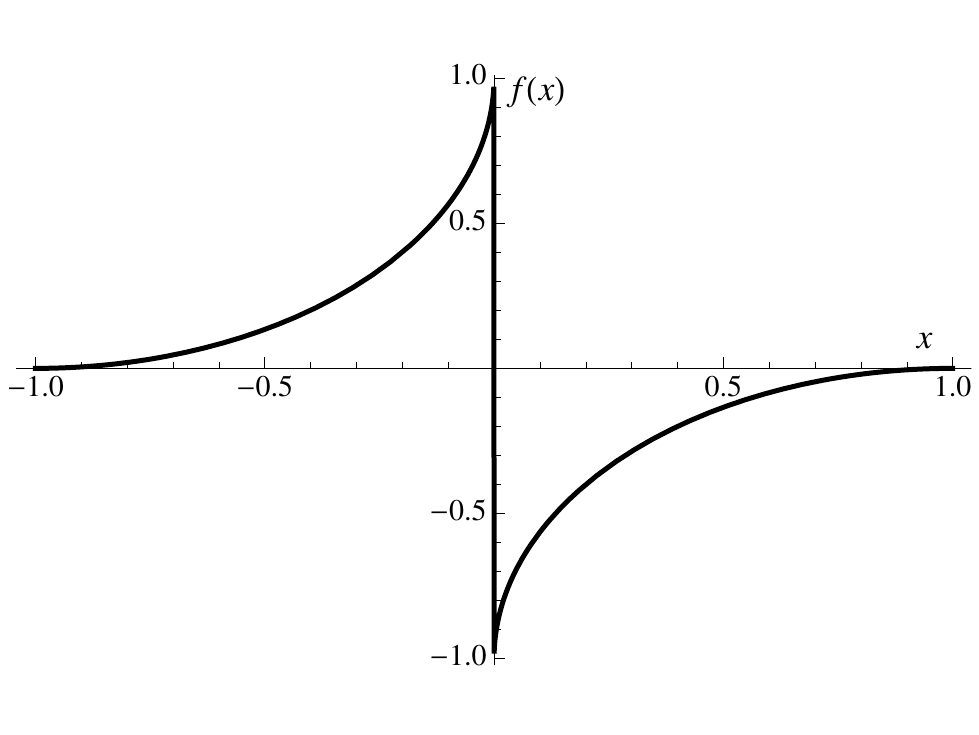}
\end{center}
\caption{The optimal map $f$ for the density of Eq.~\eqref{eq_denstriang} in the case $a=1$.}
\label{fig_ftriangN2}
\end{figure}

Taking the Gaussian
\beq
\rho(x)=(\pi)^{-1/2}e^{-x^2}
\label{eq_densgauss}
\eeq
we obtain the optimal map shown in Fig.~\ref{fig_fgauss}.

\begin{figure}[h]
\begin{center}
\includegraphics[width=8.0cm]{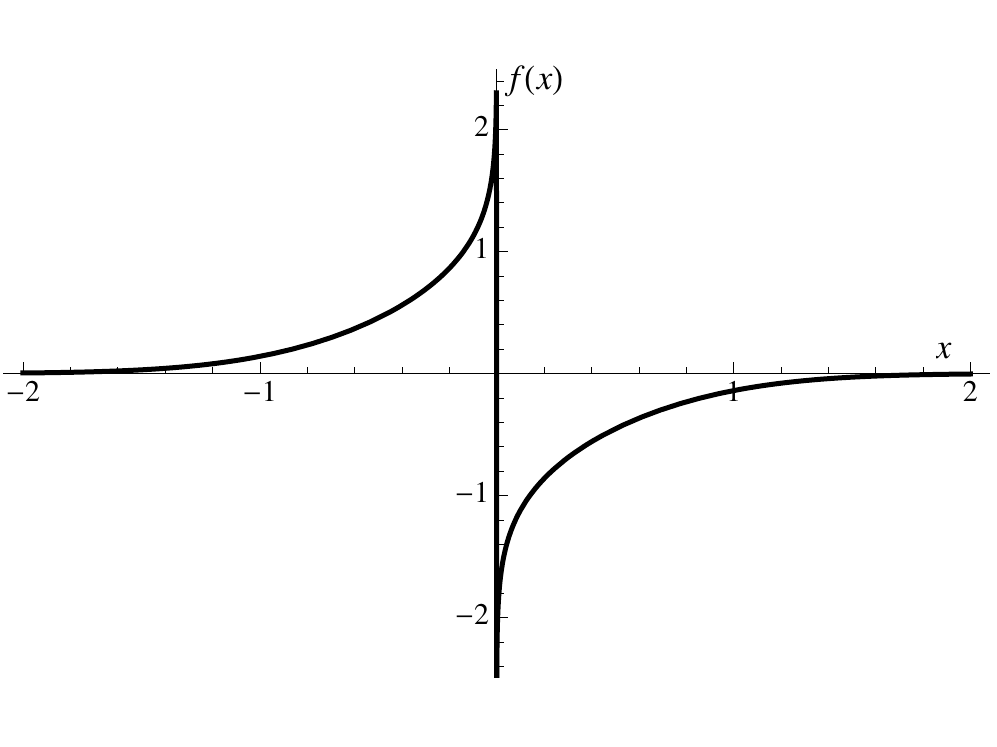}
\end{center}
\caption{The optimal map $f$ for the Gaussian density of Eq.~\eqref{eq_densgauss}.}
\label{fig_fgauss}
\end{figure}

\subsection{The case $N\ge3$ and $d=1$ dimension}\label{onedim3}

We consider the case of three particles in $\R$, with cost
$$c(x,y,z)=\frac{1}{|x-y|}+\frac{1}{|y-z|}+\frac{1}{|z-x|}.$$
The transport maps formulation aims to find two maps $f_1:\R\to\R$ and $f_2:\R\to\R$ such that $f_1^\#\rho=f_2^\#\rho=\rho$ which minimize the quantity
$$\int_\R\left(\frac{1}{|x-f_1(x)|}+\frac{1}{|f_1(x)-f_2(x)|}+\frac{1}{|f_2(x)-x|}\right)\,d\rho(x),$$
with $f_2=f_1\circ f_1$, as it follows from the indistinguishability of the three particles.

The simplest case occurs when the marginal source $\rho$ is of the form
$$\rho=\frac{1}{3}\sum_{i=1}^3\delta_{x_i}$$
in which the optimal transport maps $f_1$ are all the permutations of the points $\{x_i\}_{i=1,2,3}$ that do not send any point in itself.
In the case of a diffuse source $\rho$ we split $\rho$ into its three tertiles $\rho_1,\rho_2,\rho_3$ with $\int\rho_i\,dx=1/3$ and we send $\rho_1\to\rho_2$, $\rho_2\to\rho_3$, $\rho_3\to\rho_1$ through monotone transport maps. For instance, if $\rho$ is the Lebesgue measure on the interval $[0,1]$ we have that the optimal transport map $f_1$ is
$$f_1(x)=\left\{\begin{array}{ll}
x+1/3&\hbox{if }x\le2/3\\
x-2/3&\hbox{if }x>2/3,
\end{array}\right.$$
and correspondingly
$$f_2(x)=f_1^2(x)=\left\{\begin{array}{ll}
x+2/3&\hbox{if }x\le1/3\\
x-1/3&\hbox{if }x>1/3.
\end{array}\right.$$
\begin{figure}[h]
\begin{center}
\includegraphics[width=8.0cm]{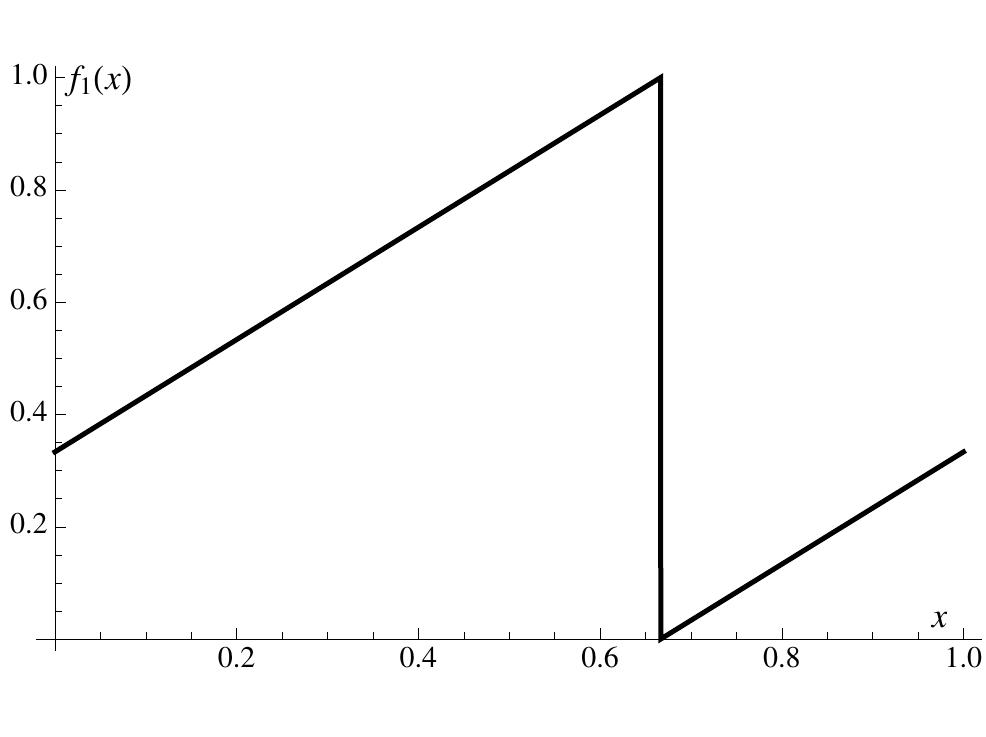}
\includegraphics[width=8.0cm]{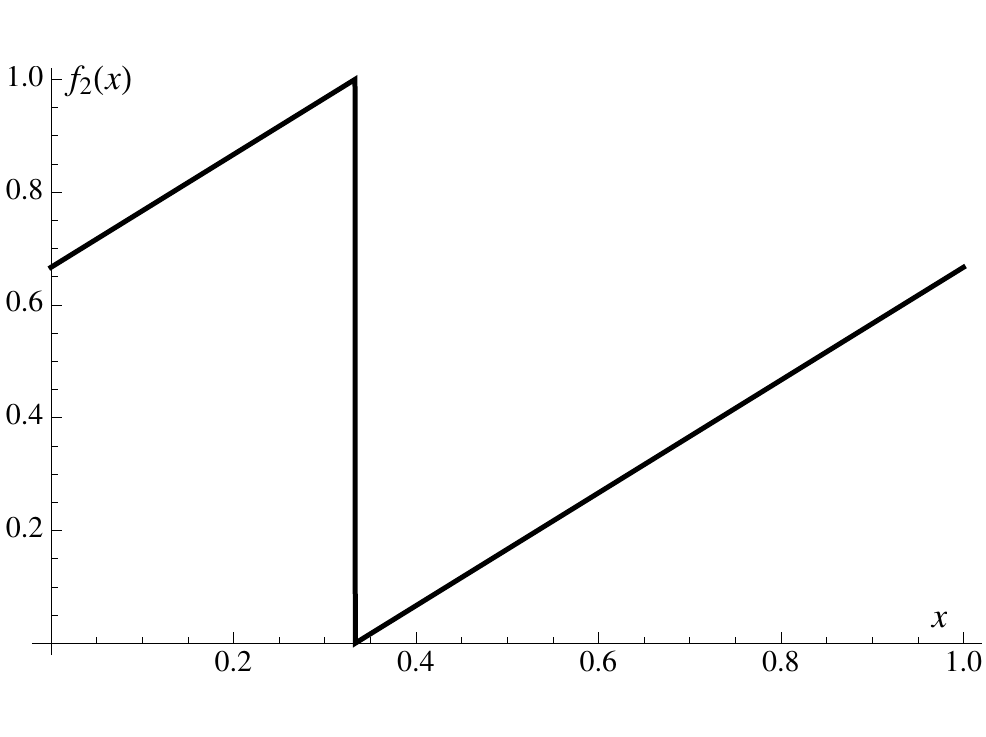}
\end{center}
\caption{The optimal maps $f_1$ and $f_2=f_1^2$ for $N=3$ and $\rho=dx$ on $[0,1]$.}
\end{figure}

Let us show that $f_1$ and $f_2$ induce an optimal plan $P$. We check the
optimality by calculating an explicit Kantorovich potential $u$ which satisfies, for all $x,y,z$:
\begin{equation}\label{ineq3part}
u(x)+u(y)+u(z)\le\frac{1}{|x-y|}+\frac{1}{|x-z|}+\frac{1}{|y-z|}
\end{equation}
and, $\forall \;x$,
\begin{eqnarray}
u(x)+u(f_1(x))+u(f_2(x)) & = & 
\frac{1}{|x-f_1(x)|}+\frac{1}{|x-f_2(x)|} \nonumber \\
& + & \frac{1}{|f_1(x)-f_2(x)|}.
\label{eq3part}
\end{eqnarray}
We remark that the right-hand side in equation \eqref{eq3part} is equal to $15/2$. To calculate $u$ we observe that the inequality \eqref{ineq3part} holds everywhere, then differentiating with respect to $x$ we obtain at a point $(x,y,z)$ of equality
$$u'(x)=-\frac{x-y}{|x-y|^3}-\frac{x-z}{|x-z|^3}\;.$$
Replacing $y$ by $f_1(x)$ and $z$ by $f_2(x)$ we obtain
\begin{equation}
u'(x)=\left\{\begin{array}{ll}
\frac{45}{4}&\hbox{ if }x\in[0,\frac{1}{3}),\\
0&\hbox{ if }x\in(\frac{1}{3},\frac{2}{3}),\\
-\frac{45}{4}&\hbox{ if }x\in(\frac{2}{3},1].
\end{array}
\right.
\end{equation}
Then we find that
\begin{equation}
u(x)=\left\{\begin{array}{ll}
\frac{45}{4}x+c&\hbox{ if }x\in[0,\frac{1}{3}],\\
\frac{15}{4}+c&\hbox{ if }x\in [\frac{1}{3},\frac{2}{3}],\\
-\frac{45}{4}x+\frac{45}{4}+c&\hbox{ if }x\in[\frac{2}{3},1].
\end{array}
\right.
\end{equation}
Equation \eqref{eq3part} gives $c=0$. By construction, then, $u$ satisfies
\eqref{eq3part} and we only need to show that it satisfies also \eqref{ineq3part}. To see this we remark that by symmetry it is enough to check the inequality in the set where $x<y<z$ and that on this set the function
$$(x,y,z)\mapsto\frac{1}{|x-y|}+\frac{1}{|x-z|}+\frac{1}{|y-z|}$$
is convex. On the other hand by the concavity of $u$ the function
$$(x,y,z)\mapsto u(x)+u(y)+u(z)$$
is concave. These two maps coincide together with their gradients at the point $(1/6,3/6,5/6)$ and then the convex one has to stay above the concave.

In the case of a possibly singular source optimal maps do not in general exist, and the optimal configurations are given by probability plans $P$. For instance, if
$$\rho=\frac{1}{4}\sum_{i=1}^4\delta_{x_i}$$
with $x_i$ ordered in an increasing way, we have that the optimal transport plan sends:
$$\begin{array}{ll}
&\delta_{x_1}+\frac{1}{3}\delta_{x_2}\to
\frac{2}{3}\delta_{x_2}+\frac{2}{3}\delta_{x_3}\\
&\frac{2}{3}\delta_{x_2}+\frac{2}{3}\delta_{x_3}\to
\frac{1}{3}\delta_{x_3}+\delta_{x_4}\\
&\frac{1}{3}\delta_{x_3}+\delta_{x_4}\to
\delta_{x_1}+\frac{1}{3}\delta_{x_2}.
\end{array}$$

When $N\ge4$ similar arguments as above can be developed, giving transport maps $f_1$, $f_1^2$, \dots, $f_1^{N-1}$ that minimize the total cost
$$\int c\big(x,f_1(x),f_1^2(x),\dots,f_1^{N-1}(x)\big)\,d\rho(x),$$
where $c(x_1,\dots,x_N)$ is given in Eq.~\eqref{eq_Coulcost}. Some of these results were also obtained by Seidl \cite{Sei-PRA-99}, again using physical arguments.

\section{Conclusions and perspectives}\label{sec_conc}
We have shown that the strong-interaction limit of electronic density functional theory can be rewritten as a mass transportation theory problem, thus creating a link between two different, well established, research areas. This is already interesting {\it per se}: it allows to import and generalize results from one domain to the other. In particular, with our reformulation we were able to prove immediately several results on the strong-interaction limit of DFT. Even more interesting, we could show that the problem of finding the minimum interaction energy in a given density can be rewritten in a convenient dual form consisting of a minimization under linear constraints, paving the way to the use of linear programming techniques to solve the strictly-correlated electron problem.

Dual reformulations have been already proved very useful in the context of electronic structure calculations: for example, in Ref.~\onlinecite{CanStoLew-JCP-06} the solution of the physical hamiltonian by  optimizing the second-order reduced density matrix has been tackled with a suitable dual problem. The use of Legendre transform techniques for the simplification of
minimizations involving permutations in the many electron problem has also been stressed and applied in Refs.~\onlinecite{OsbLev-PRA-86} and \onlinecite{OsbLev-PRA-87}, with very interesting results. All these approaches focused on the quantum mechanical problem, while here we deal with a special problem that is essentially of classical nature, but contains quantum-mechanical information via the electronic density. We know now that the optimal transport formulation is the right mathematical framework for the strong-interaction limit of density functional theory.

It is also worth to mention that the formalism developed here can be of interest for approaches to the many-electron problem which use $k$-electron distribution functions (i.e., the diagonal of the $k^{\rm th}$ order reduced density matrix), such as those of Refs.~\onlinecite{Aye-PRA-06} and \onlinecite{Liu-THESIS-06}. In fact, in these approaches one usually constructs a $k$-electron distribution function $\rho_k(\rv_1,\dots,\rv_k)$ with a given density, possibly minimizing the electron-electron repulsion energy. This would result in the same Kantorovich formulation considered here.

 The formal and practical aspects of our new reformulation are enticing for DFT: making routinely available the piece of exact information contained in the strong-interaction limit can largely broaden its applicability, both by developing a ``SCE DFT'' \cite{GorSeiVig-PRL-09,LiuBur-JCP-09,GorSei-PCCP-10} (which uses a strong interacting system as a reference), and via new exchange-correlation functionals for standard KS DFT \cite{SeiPerKur-PRL-00,GorSei-PCCP-10,Bec-cecam-11}.
Future work will be devoted to exploit the practical aspects of this reformulation.

{\em Note added in proof:} While this article was in review, we become aware that a related work  \cite{CotFriKlu-arxiv-11}
was posted on arXiv. In  \cite{CotFriKlu-arxiv-11} the particle-particle interaction term is
minimized by only considering the pair density. By neglecting the
$N$-representability issue, this leads to a two-particle problem with
only one map (or co-motion function).

\section*{Acknowledgments}
This work was supported by the Netherlands Organization for Scientific Research (NWO) through a Vidi grant.
PG-G thanks Roland Assaraf for suggesting to read about mass transportation theory, Giovanni Vignale and Michael Seidl for useful discussions, and Andr\'e Mirtschink for a critical reading of the manuscript. 

\appendix
\section{Proof of Theorem \ref{thsumm}}
\label{app_proof}
For the sake of simplicity we give a sketch of the proof only in the case of two particles, the general case can be obtained in a similar way. We thus consider the case $N=2$ in
$\R^d$ with two equal marginals, $p_1=p_2=p$ with $p\in P(\R^d)$. The problem is then
\begin{equation}\label{cost2}
\min\Big\{\int\frac{1}{|x-y|}\,dP(x,y)\ :\ \pi_j^\# P=p_j,\hbox{ for }j=1,2\Big\},
\end{equation}
and $p$ will be assumed absolutely continuous i.e.~of the form $\rho(x)\,dx$ with $0\le\rho(x)$ and $\int\rho(x)\,dx=1$.

By definition, the Kantorovich potential $u$ is a maximizer for the dual problem according to Theorem \ref{dualbis} and Remark \ref{remsim}. If we denote by $P$ an optimal plan of transport (which exists by Theorem \ref{exis}) then in the case $N=2$ considered here $u$ maximizes the functional
\begin{equation}\label{dualefisico}
\int_{\R^d\times\R^d}\big(u(x)+u(y)\big)\,dP(x,y)
\end{equation}
among all the functions which satisfies the constraint
\begin{equation}\label{vincolonewton}
u(x)+u(y)\le\frac{1}{|x-y|}.
\end{equation}
Under the current assumptions such a maximizer exists by Theorem \ref{dualbis} above, but $u$ is only a Borel function which takes values in $[-\infty,+\infty)$. Much more then is needed to carry on the necessary computations and we will deduce the needed properties. The proof will be made of several steps. But first, let us fix the following notation: for a transport plan $P$ we denote by $\spt(P)$ the support of $P$, i.e. the smallest closed set $F$ such that $P(\R^d\times\R^d\setminus F)=0$.

\bigskip\underline{\bf Step 1} 
The first step is the following intuitive fact about optimal transport plans. If $\Po$ is an optimal transport plan, then
$$0<|x-y|\qquad\forall(x,y)\in\spt(\Po).$$
Indeed, if by contradiction a point $(\overline x,\overline x)\in\spt(\Po)$ we may find a better transport plan $\tilde P$ by exchanging the mass around $(\overline x,\overline x)$ with the one around another point $(\tilde x,\tilde y)\in\spt(\Po)$ having $\tilde x\neq\tilde y\neq\overline x$.

\bigskip\underline{\bf Step 2} Actually something more can be said. Let $\Po$ be an optimal transport plan; then for all $R>0$ there exists $\alpha(R)>0$ such that
$$\alpha(R)<|x-y|\qquad\forall x\in B(0,R),\ \forall(x,y)\in\spt(P).$$
Indeed let $x\in B(0,R)$ and $(x,y)\in\spt(\Po)$; by the point above and by compactness, the diagonal and the support of $\Po$ have positive distance in the set $\overline{B(0,R)}\times\overline{B(0,2R)}$ and we denote by $\beta(R)$ such a distance. It follows that
$$\min\{\beta(R),R\}\le|x-y|.$$
We then define $\alpha(R):=\min\{\beta(R),R\}$. Moreover, we may choose the function $\alpha$ non increasing. 

\bigskip\underline{\bf Step 3} Using the second step we now prove that there are Kantorovich potentials which are bounded. First we remark that we can choose a Kantorovich potential $v$ which satisfies
\begin{equation}\label{hopflax}
v(x)=\inf_{y\in\R^d}\Big\{\frac{1}{|x-y|}-v(y)\Big\}.
\end{equation}
We start with a potential $u$ and we notice that by definition 
$$u(x)\le\inf_{y\in\R^d}\Big\{\frac{1}{|x-y|}-u(y)\Big\}.$$
Then we can consider
$$\overline u(x)=\inf_{y\in\R^d}\Big\{\frac{1}{|x-y|}-u(y)\Big\}.$$
Clearly $u\le\overline u$. Even if $\overline u$ does not satisfy the constraint, from the definition we get
$$\overline u(x)+\overline u(y)=\inf_{z\in\R^d}\Big\{\frac{1}{|x-z|}-u(z)\Big\}+\inf_{z\in\R^d}\Big\{\frac{1}{|z-y|}-u(z)\Big\},$$
and taking $y$ as test in the first term of the right-hand-side and $x$ in the second it follows that
$$\overline u(x)+\overline u(y)\le\frac{2}{|x-y|}-u(y)-u(x),$$
or equivalently if we define $\tilde u(x)=2^{-1}\big(u(x)+\overline u(x)\big)$
$$u(x)\le\tilde u(x)\le\overline u(x)\quad\mbox{and}\quad\tilde u(x)+\tilde u(y)\le\frac{1}{|x-y|}.$$
We may now define
\begin{eqnarray}
v(x) & = & \sup\big\{w(x)\ :\ u(x)\le w(x)\le\overline u(x) \nonumber \\ 
& & \mbox{ and }w\mbox{ satisfies }\eqref{vincolonewton}\big\}. \nonumber
\end{eqnarray}
The function $v(x)$ clearly satisfies \eqref{vincolonewton}, and if $v\neq\overline v$ since
\begin{eqnarray}
\overline v(x)& = & \inf_{y\in\R^d}\Big\{\frac{1}{|x-y|}-v(y)\Big\} \nonumber \\
& \le & \inf_{y\in\R^d}\Big\{\frac{1}{|x-y|}-u(y)\Big\}=\overline u(x) \nonumber
\end{eqnarray}
then $v< \tilde v\le\overline u$ which contradicts the maximality of $v$. Finally, $v$ maximizes the cost \eqref{dualefisico} since $u\le v$.

\bigskip\underline{\bf Step 4} As anticipated in Theorem \ref{thsumm}
if $v$ is a Kantorovich potential which satisfies \eqref{hopflax}, 
then there exists a costant $C$ such that $|v|\le C$. Let $\Po$ be an optimal plan of transport. The condition
$$\int_{\R^d\times\R^d}\big(v(x)+v(y)\big)\,d\Po(x,y)
=\int_{\R^d\times\R^d}\frac{d\Po(x,y)}{|x-y|}$$
together with the condition
$$v(x)+v(y)\le\frac{1}{|x-y|}$$
implies that
$$v(x)+v(y)=\frac{1}{|x-y|}\mbox{ for $\Po$-a.e. $x$ and }y,$$
and then in particular $v$ is finite $\rho$-a.e. Moreover, setting
$$G=\Big\{x\ :\ -\infty<v(x)\mbox{ and }\exists\, y\mbox{ s.t. }v(x)+v(y)=\frac{1}{|x-y|}\Big\}$$
it follows from the discussion above that $\rho(G)=1$. Let $\overline x\in G$ be a point of density 1 for $G$ and let $\alpha$ and $r$ be such that $\alpha>r$ and
\begin{enumerate}
\item for all $s\le r$ we have $|B(\overline{x},s)\cap G|/|B(\overline{x},s)|\ge3/4$,
\item for all $x\in B(\overline x,r)$ if $(x,y)\in\spt(P)$ then $\alpha<|x-y|$.
\end{enumerate}
Setting $L=v(\overline x)$ we have that for every $z\in\R^d\setminus B(\overline x,r/4)$
\begin{equation}\label{above1}
v(z)=\inf_{y\in\R^d}\Big\{\frac{1}{|y-z|}-v(y)\Big\}\le\frac{4}{r}-L.
\end{equation}
On the other hand for every $z\in B(\overline x,r/4)\cap G$ there exists $y$ such that $\alpha<|z-y|$, $(z,y)\in\spt(\Po)$ and $v(z)+v(y)=1/|z-y|$. Then
\begin{eqnarray}
v(z)& = & \inf_{y\in\R^d}\Big\{\frac{1}{|y-z|}-v(y)\Big\}\nonumber \\
&\le &
\inf_{|y-z|\ge\alpha}\Big\{\frac{1}{|y-z|}-v(y)\Big\}=v(z) \nonumber
\end{eqnarray}
since $r<\alpha$, $\R^d\setminus B(z,\alpha)\subset\R^d\setminus B(\overline x,r/4)$, and then from the estimate \eqref{above1}
\begin{equation}\label{below1} 
v(z)=\inf_{|y-z|\ge\alpha}\Big\{\frac{1}{|y-z|}-v(y)\Big\}\ge L-\frac{1}{r}.
\end{equation}
To get a control of $v$ from above in $B(\overline x,r/4)$ we observe that if $\lambda\le r4^{-1/d}$ and $z\in B(\overline x,r/4)$, since
$$|B(z,\lambda)|=\omega_d\lambda^d\le\frac{\omega_d}{4}r^d\;,$$
then there exists at least one $y_z\in G\cap B(\overline x,r/4)\setminus B(z,\lambda)$ such that, from the estimate \eqref{below1},
\begin{eqnarray}
v(z) &\le & \inf_{|y-z|\ge\lambda}\Big\{\frac{1}{|y-z|}-v(y)\Big\} \nonumber \\
& \le & \frac{1}{|y_z-z|}-v(y_z)\le\frac{1}{\lambda}+\frac{1}{r}-L.
\label{above2}
\end{eqnarray}
Estimates \eqref{above1} and \eqref{above2} give a bound from above on $v$ by a constant $K$. The estimate from below is now straightforward since
$$v(x)=\inf_{\R^d}\Big\{\frac{1}{|x-y|}-v(y)\Big\}\ge-K.$$

\bigskip\underline{\bf Step 5} The previous steps permit to gain more regularity on the potential $v$. Let indeed $v$ be a Kantorovich potential which satisfies \eqref{hopflax}; we show that $v$ is differentiable almost everywhere. To see this we consider the family of functions
$$v_n(x)=\inf_{\alpha(n)<|x-y|}\Big\{\frac{1}{|x-y|}-v(y)\Big\}.$$
Since $\alpha$ is nonincreasing we have
$$v_{n+1}(x)\le v_n(x).$$
Moreover each $v_n$ is a Lipschitz function of Lipschitz constant $1/\alpha^2(n)$. By Step 2 for $x\in G$ if $|x|<m<n$ then $v(x)=v_n(x)=v_m(x)$. Then on $G$ the potential $v$ coincide locally with a Lipschitz function which is well known to be differentiable almost everywhere.

\section{Proof of the existence of an optimal transport map for $N=2$}
\label{app_fN2}
Once the existence of an a.e.~differentiable Kantorovich potential $v$ is established, we may consider the problem of showing the existence of an optimal transport map ({\it co-motion function}) $f$. In the case $N=2$ the proof can be achieved by using the basic idea of differentiating inequality \eqref{vincolonewton} at the points of equality.

Let $\Po$ be an optimal transport plan and let $G$ be defined as above. If $G_N:=G\cap B(0,N)$ we prove that for almost every $x\in G_N$ there exists only one $y$ such that $(x,y)\in\spt(\Po)$ and we give an explicit expression for such $y$. It follows that $\Po$ is induced by an optimal transport map. Let $v_N$ be the function defined above; since $v_N$ coincides with $v$ on $G_N$, for every $x\in G_N$ and $y$ the inequality
$$v_N(x)+v(y)\le\frac{1}{|x-y|}$$
holds. Since $\Po$ is an optimal transport plan and $\rho=a(x)\,dx$, then for $\Po$-a.e. $(x,y)\in\spt(\Po)$, $x$ belongs to $G_N$ for a suitable $N$, $x$ is a density point for $G_N$ and $v_N$ is differentiable at $x$. Since for $z\in G_N$
$$v_N(z)\le\frac{1}{|z-y|}-v(y)$$
and equality holds for $z=x$ then if we differentiate the functions $v_N$ and $\psi(z)=\frac{1}{|z-y|}-v(y)$ we obtain
$$\nabla v_N(x)=-\frac{x-y}{|x-y|^3}$$
from which it follows
\begin{equation}\label{inverse}
y=x+\frac{1}{|\nabla v_N (x)|^{3/2}}\nabla v_N (x).
\end{equation}
From equation \eqref{inverse} we deduce that for $\Po$-a.e. $(x,y)$ the point $y$ is uniquely determined by $x$ and this concludes the proof of the existence of an optimal transport map $f$, by defining
$$f(x)=x+\frac{1}{|\nabla v_N (x)|^{3/2}}\nabla v_N (x)$$
whenever $x\in B(0,N)$.


\end{document}